\documentclass[12pt]{article}

\usepackage[paper=a4paper,dvips,top=1.5in,left=1.5in,right=1.5in, foot=1cm,bottom=1.8in]{geometry}

\usepackage{amsmath, amssymb, amsfonts, amsthm}

\usepackage[T1]{fontenc}

\usepackage[ruled, vlined]{algorithm2e}
\usepackage{graphics}

\def\Exp{\qopname\relax m{E}}

\newcommand{\runtime}[3]%
{$O^{*}(#1^{#3})=O^{*}(2^{#2\cdot #3})$}

\def\Violators{{\sf V}}
\def\Extremes{{\sf X}}

\newcommand{\Proof}[1]%
	{%
	\noindent %
	\textbf{Proof} #1.~%
	}%
\newsavebox{\smallProofsym}
\savebox{\smallProofsym}
	{
	\begin{picture}(7,7)%
	\put(0,0){\framebox(6,6){}}%
	\put(0,2){\framebox(4,4){}}%
	\end{picture}%
	}
\newcommand{\smallpop}[1]%
	{%
	\mbox{} \hfill #1~~\usebox{\smallProofsym}\!\!\!\!\!\!\ %
	}%
	{%
	\Proof{#1}}{\smallpop{}%
	\medskip%
	}%

\newsavebox{\smallEndsym}%
\savebox{\smallEndsym}%
	{%
	\begin{footnotesize}%
	$\square$%
	\end{footnotesize}%
	}
\newcommand{\smalleop}[1]%
	{%
	\mbox{} \hfill #1~~\usebox{\smallEndsym}\!\!\!\!\!\!\ %
	}
\newsavebox{\smallDefsym}
\savebox{\smallDefsym}%
	{%
	\begin{footnotesize}%
	$\blacksquare$%
	\end{footnotesize}%
	}
\newcommand{\smalldop}[1]%
	{%
	\mbox{} \hfill #1~~\usebox{\smallDefsym}\!\!\!\!\!\!\ %
	}

\newtheorem{theorem2}{Theorem}[section]
\newenvironment{theorem}{\begin{theorem2}}{\end{theorem2}}
\newtheorem{lemma2}[theorem2]{Lemma}
\newenvironment{lemma}{\begin{lemma2}}{\end{lemma2}}
\newtheorem{corollary2}[theorem2]{Corollary}
\newenvironment{corollary}{\begin{corollary2}}{\end{corollary2}}
\newtheorem{definition2}[theorem2]{Definition}
\newenvironment{definition}{\begin{definition2}}{\end{definition2}}
\newtheorem{claim2}[theorem2]{Claim}
\newenvironment{claim}{\begin{claim2}}{\end{claim2}}
\theoremstyle{remark}

\newtheorem{remark2}{Remark}[section]

\newtheorem{example2}{Example}[section]

\let\newEnumerate=\enumerate
\renewenvironment{enumerate}{\newEnumerate\itemsep=0.0cm}{\endlist}
\let\newItemize=\itemize
\renewenvironment{itemize}{\newItemize\itemsep=0.0cm}{\endlist}

\newcommand{\ngt}[1]		
	{\overline{#1}}

	\newcommand{\SAT}		
		{{\mathrm{sat}}}
	\newcommand{\SAt}		
		{{\mathrm{sa{\mbox{\footnotesize $\mathrm{t}$}}}}}

	\newcommand{\UNSAT}		
		{\overline{\SAt}}
\newcommand{\Unsat}[2]		
	{\UNSAT_{#1}(#2)}

\renewcommand{\arraystretch}{1.15}

\def\stageI{\texttt{GA}}
\def\stageII{\texttt{SA}}
\def\stageIIforever{\texttt{SA\_forever}}
\def\stageIII{\texttt{BFA}}

\begin{document}

\date{June 25, 2009}
\title{Clarkson's Algorithm for Violator Spaces}
\author{Yves Brise, Bernd G\"artner\\
{\small Swiss Federal Institute of Techology (ETHZ)}\\{\small 8092 Zurich, Switzerland}\\{\small 
\texttt{ybrise | gaertner@inf.ethz.ch}}\\\\
}

\maketitle
\begin{center}
	\textbf{Abstract}\\
\end{center}
Clarkson's algorithm is a two-staged randomized algorithm for solving linear programs. 
This algorithm has been simplified and adapted to fit the framework of LP-type problems. In this
framework we can tackle a number of non-linear problems such as computing the smallest
enclosing ball of a set of points in $\mathbb{R}^d$.
In 2006, it has been shown that the algorithm in its original form works for violator spaces too, which 
are a proper
generalization of LP-type problems. It was not clear, however, whether previous simplifications of the 
algorithm carry over to the new setting.

In this paper we show the following theoretical results: $(a)$ It is shown, for the first time, that Clarkson's
second stage can be simplified.  $(b)$ The previous simplifications of Clarkson's first stage carry
over to the violator space setting.  
$(c)$ Furthermore, we show the equivalence of violator spaces and partitions of the hypercube by hypercubes.

~\\
\textbf{Keywords:} Clarkson's Algorithm, Violator Space, LP-type Problem, Hypercube Partition

\section{Introduction}
\paragraph{Clarkson's algorithm.}
Clarkson's randomized algorithm \cite{c-lvali-95} is the earliest
practical linear-time algorithm for linear programming with a fixed
number of variables.  Combined with a later algorithm by Matou{\v{s}}ek,
Sharir and Welzl \cite{MSW}, it yields the best (expected) worst-case
bound in the unit cost model that is known today.  The combined
algorithm can solve any linear program with $d$ variables and $n$
constraints with an expected number of $O(d^2 n + \exp(O(\sqrt{d\log
  d})))$ arithmetic operations \cite{GaerWel1}.

Clarkson's algorithm consists of two primary stages, and it requires
as a third stage an algorithm for solving small linear programs with
$O(d^2)$ constraints. The first two stages are purely combinatorial and use
very little problem-specific structure.  Consequently, they smoothly
extend to the larger class of \emph{LP-type problems} \cite{MSW},
with the same running time bound as above for concrete problems in this
class, like finding the smallest enclosing ball of a set of $n$ points
in dimension $d$ \cite{GaerWel1}.

Both primary stages of Clarkson's algorithm are based on random sampling and
are conceptually very simple. The main idea behind the use of randomness
is that we can solve a subproblem subject to only a small number of
(randomly chosen) constraints, but still have only few (of all) constraints
that are violated by the solution of the subproblem.
 However, some extra machinery was
originally needed to make the analysis go through.  More precisely,
in both stages there needed to be a check that the each individual
random choice was ``good'' in a certain sense. Then in the analysis
one needed to make the argument that the bad cases do not
occur too often. For the first
stage, it was already shown by G\"artner and Welzl that these extra checks
 can be removed \cite{GWSampl01}. The result is what we call
the \emph{German Algorithm} below. In this paper, we do the removal
also for the second stage, resulting in the \emph{Swiss Algorithm}.
(The names come from certain aspects of German and Swiss mentality
that are reflected in the respective algorithms.) We believe that the
German and the Swiss Algorithm together represent the essence of
Clarkson's approach.

\paragraph{Violator spaces.}
G\"artner, Matou\v{s}ek, R\"ust, and \v{S}kovro\v{n} proved that
Clarkson's original algorithm is applicable in a still broader setting
than that of LP-type problems: It actually works for the class of
\emph{violator spaces} \cite{journals/dam/GartnerMRS08}. At first
glance, this seems to be yet another generalization to yet another
abstract problem class, but as \v{S}kovro\v{n} has shown, it stops
here: the class of violator spaces is the most general one for which
Clarkson's algorithm is still guaranteed to work \cite{skovronP}. In a
nutshell, the difference between LP-type problems and violator spaces
is that for the latter, the following trivial algorithm may cycle even
in the nondegenerate case: maintain the optimal solution subject to a
subset $B$ of the constraints; as long as there is some constraint $h$
that is violated by this solution, replace the current solution by the
optimal solution subject to $B\cup\{h\}$, and repeat. Examples of
such \emph{cyclic} violator spaces can be found in \cite{skovronP}.
For a very easy and intuitive example see also  \cite{journals/dam/GartnerMRS08}.

It was unknown whether the analysis of the German Algorithm (the
stripped-down version of Clarkson's first stage) also works for
violator spaces. For LP-type problems, the analysis is nontrivial and
constructs a ``composite'' LP-type problem. Here we show that this can
still be done for violator spaces, in essentially the same way.

For the Swiss Algorithm (the stripped-down version of Clarkson's
second stage), we provide the first analysis at all. The fact that it
works in the fully general setting of violator spaces comes naturally.

The main difference of the German and the Swiss algorithm compared
to their original formulations is the following. In both stages, at some point,
Clarkson's algorithm checks how many violated constraints some random
sample of constraints produces.
If there are too many, then the algorithm discards the sample and resamples.
The reason for this is that the analysis requires a bound on the number of
violators in each step. We essentially show that this bound only needs to hold
in expectation (and does so) for the analysis to go through.
So the  checks that we mentioned before are only an analytical tool, and not necessary for
the algorithms to work.

Let us point out that no subexponential
algorithm for finding the basis (i.e. ``solution'') of a violator space is known. Therefore,
we can only employ brute force to ``solve'' small violator spaces. Note that,
e.g., in the context of linear programming, finding a basis means identifying the
constraints which are tight at an optimal point. We
call this the Brute Force Algorithm (\stageIII). Hence, the resulting best
worst-case bound known degrades to $O(d^2n + f(d))$, where $f$ is
some exponential function of $d$. In this paper, we will not investigate
this point further and use \stageIII{ }as a black box.

\paragraph{The German Algorithm (\stageI).}
Let us explain the algorithm for the problem of finding the smallest
enclosing ball of a set of $n$ points in $\mathbb{R}^d$ (this problem fits
into the violator space framework). The algorithm proceeds in rounds
and maintains a working set $G$, initialized with a subset $R$ of $r$
points drawn at random. In each round, the smallest enclosing ball of
$G$ is being computed (by some other algorithm). For the next round,
the points that are unhappy with this ball (the ones that are outside)
are being added to $G$. The algorithm terminates as soon as everybody
is happy with the smallest enclosing ball of $G$.

The crucial fact that we reprove below in the violator space framework
is this: the number of rounds is at most $d+2$, and for $r\approx
d\sqrt{n}$, the expected maximum size of $G$ is bounded by
$O(d\sqrt{n})$. This means that \stageI{ }reduces a
problem of size $n$ to $d+2$ problems of expected size $O(d\sqrt{n})$.
We call this the German Algorithm, because it takes -- typically
German -- one decision in the beginning which is then efficiently
pulled through.

\paragraph{The Swiss Algorithm (\stageII).} 
Like \stageI, this algorithm proceeds in rounds, but it
maintains a voting box that initially contains one slip per point. In
each round, a set of $r$ slips is drawn at random from the voting box,
and the smallest enclosing ball of the corresponding set $R$ is
computed (by some other algorithm). For the next round, all slips are
put back, and on top of that, the slips of the unhappy points are being
doubled. The algorithm terminates as soon as everybody is happy with
the smallest enclosing ball of the sample $R$.

Below, we will prove the following: if $r\approx d^2$, the expected number
of rounds is $O(\log n)$. This means that \stageII{ }reduces a
problem of size $n$ to $O(\log n)$ problems of size $O(d^2)$. We call
this the Swiss Algorithm, because it takes -- typically Swiss -- many 
independent local decisions that magically fit together in the end.

\paragraph{Hypercube partitions.}
A hypercube partition is a partition of the vertices of the hypercube
such that every element of the partition is the set of vertices of
some subcube. It was known that every \emph{nondegenerate} violator
space induces a hypercube partition
\cite{SkovronMatousekLPdim,MatousekLPdim}.  We prove here that also
the converse is true, meaning that we obtain an alternative
characterization of the class of violator spaces. While this result is
not hard to obtain, it may be useful in the future for the problem of
counting violator spaces. Here, the initial bounds provided by
\v{S}kovro\v{n} are still the best known ones \cite{skovronP}.

\paragraph{Applications.} 
We would love to present a number of convincing
applications of the violator space framework, and in particular of the
German and the Swiss Algorithm for violator spaces. Unfortunately, we
cannot. There is one known application of Clarkson's algorithm that
really requires it to work for violator spaces and not just LP-type
problems \cite{journals/dam/GartnerMRS08}; this application (solving
generalized $P$-matrix linear complementarity problems with a fixed
number of blocks) benefits from our improvements in the
sense that now also the German and the Swiss Algorithm are applicable
to it (with less random resources than Clarkson's algorithm).

Our main contributions are therefore theoretical: we show that
Clarkson's second stage can be simplified (resulting in the Swiss
Algorithm), and this result is new even for LP-type problems and
linear programming. The
fact that Clarkson's first stage can be simplified (resulting in the
German Algorithm) was known for LP-type problems; we extend it to
violator spaces, allowing the German Algorithm to be used for solving
generalized $P$-matrix linear complementarity problems with a fixed
number of blocks.

We  believe that our results are significant
contributions to the theory of abstract optimization frameworks
themselves. We have now arrived at a point where Clarkson's algorithm
has been shown to work in the most general abstract setting that is
possible, and in probably the most simple variant that can still
successfully be analyzed.

\section{Prerequisites}

\subsection{The Sampling Lemma}
The following lemma is due to G\"artner and Welzl in 
\cite{GWSampl01} and was adapted to violator
spaces in \cite{journals/dam/GartnerMRS08}. We repeat it here for the sake
of completeness, and because its proof and formulation are very
concise.
Let $S$ be a set of size $n$, and $\varphi: 2^S\rightarrow \mathbb{R}$  a function that maps any
set $R\subseteq S$ to some value $\varphi(R)$. Define
\begin{eqnarray}
	\Violators(R) &:=& \{ s\in S\backslash R \,|\, \varphi(R\cup \{s\}) \neq \varphi(R)\}, \label{def:violator}\\
	\Extremes(R) &:=& \{s\in R \,|\, \varphi(R\backslash\{s\})\neq \varphi(R)\}\label{def:extreme}.
\end{eqnarray}
$\Violators(R)$ is the set of \emph{violators} of $R$, while
$\Extremes(R)$ is the set of \emph{extreme elements} in $R$. Obviously,
\begin{equation}
	s \mbox{ violates } R \Leftrightarrow s \mbox{ is extreme in } R\cup\{s\}.\label{eqViolators}
\end{equation}

For a random sample $R$ of size $r$, i.e., a set $R$ chosen uniformly at random from the set ${S 
\choose r}$ of all $r$-element subsets of $S$, we define random variables $\Violators_r: R \mapsto |
\Violators(R)|$ and
$\Extremes_r: R \mapsto |\Extremes(R)|$, and we consider the expected values
\begin{eqnarray*}
	v_r &:=& \Exp[\Violators_r],\\
	x_r &:=& \Exp[\Extremes_r].
\end{eqnarray*}

\begin{lemma}[Sampling Lemma, \cite{GWSampl01, journals/dam/GartnerMRS08}]
\label{lemma:sampling}
For $0\leq r < n$,
\[
	\frac{v_r}{n-r} = \frac{x_{r+1}}{r+1}.
\]
\end{lemma}
\begin{proof}
	Using the definitions of $v_r$ and $x_{r+1}$ as well as (\ref{eqViolators}), we can argue as 
follows:

\begin{center}
\begin{math}
\renewcommand{\arraystretch}{2.15}
\displaystyle
\begin{array}{lll}
\displaystyle
	{n \choose r} v_r &=& \displaystyle\sum_{R\in{S\choose r}} \sum_{s\in S\backslash R}[s \mbox{ violates } R] \\
	&=& \displaystyle \sum_{R\in{S\choose r}} \sum_{s\in S\backslash R} [s \mbox{ is extreme in } R\cup \{s\}]\\
	&=& \displaystyle \sum_{Q\in {S \choose r+1}} \sum_{s\in Q} [s \mbox{ is extreme in } Q] \\
	&=& \displaystyle{n \choose r+1}x_{r+1}.
\end{array}
\end{math}
\end{center}

Here, $[\cdot]$ is the indicator variable for the event in brackets. Finally, ${n \choose r+1}/{n \choose 
r} = (n-r)/(r+1)$.
\end{proof}

\subsection{Violator Spaces}

\begin{definition}
\label{def:violatorspace}
A \emph{violator space} is a pair $(H,\Violators)$, where $H$ is a finite set and $\Violators$
is a mapping $2^{H}\rightarrow 2^{H}$ such that the following two conditions are fulfilled.

\vspace{0.2cm}
\begin{tabular}{ll}
	Consistency: & $G\cap \Violators(G) = \emptyset$ holds for all $G\subseteq H$, and \\
	Locality: & for all $F\subseteq G \subseteq H$, where $G\cap \Violators(F) = \emptyset$,\\
	& we have $\Violators(G) = \Violators(F)$.
\end{tabular}  
\end{definition}

\begin{lemma}[Lemma 17, \cite{journals/dam/GartnerMRS08}]
\label{lemma:monotonicity}
Any violator space $(H,\Violators)$ satisfies \emph{monotonicity} defined as follows:

\vspace{0.2cm}
\begin{tabular}{ll}
	Monotonicity: & $\Violators(F) = \Violators(G)$ implies $\Violators(E) = \Violators(F)=
\Violators(G)$\\
	& for all sets $F\subseteq E \subseteq G \subseteq H$.
\end{tabular}  
\end{lemma}
\begin{proof}
Assume $\Violators(E)\neq\Violators(F),\Violators(G)$. Then locality yields
$\emptyset\neq E\cap\Violators(F)=E\cap \Violators(G)$ which contradicts consistency.
\end{proof}

\begin{definition}
Consider a violator space $(H,\Violators)$.
\begin{itemize}
\item[(i)]We say that $B\subseteq H$ is a \emph{basis}
if for all proper subsets $F\subset B$ we have $B\cap\Violators(F)\neq \emptyset$. For $G\subseteq 
H$, a basis of $G$ is a minimal subset $B$ of $G$ with $\Violators(B)=\Violators(G)$.
A basis in $(H,\Violators)$ is a basis of some set $G\subseteq H$.
\item[(ii)] The \emph{combinatorial dimension} of $(H,\Violators)$, denoted by $\dim(H,\Violators)$,
is the size of the largest basis in $(H,\Violators)$.
\item[(iii)] $(H,\Violators)$ is \emph{nondegenerate} if every set set $G\subseteq H$, $|G|\geq \dim(H,
\Violators)$,
has a unique basis. Otherwise $(H,\Violators)$ is \emph{degenerate}.
\end{itemize}
\end{definition}

Observe that a minimal subset $B\subseteq G$ with $\Violators(B)=\Violators(G)$ is indeed a basis:
Assume for contradiction that there is a set $F\subset B$ such that $B\cap\Violators(F)=\emptyset$.
Locality then yields $\Violators(B)=\Violators(F)=\Violators(G)$, which contradicts the minimality of $B$.
Also, note that, because of consistency, any basis $B$ of $H$ has no violators $\Violators(H)=\Violators(B) = \emptyset$.

\begin{corollary}[of Lemma \ref{lemma:sampling}]
Let $(H,\Violators)$ be a violator space of combinatorial dimension $d$, and $|H|=n$. If we choose
a subset $R\subseteq H$, $|R|=r \leq n$, uniformly at random, then
\[
	\Exp[|\Violators(R)|] \leq d \frac{n-r}{r+1}.
\]
\end{corollary}
\begin{proof}
The corollary follows from the Sampling Lemma \ref{lemma:sampling}, with the observation
that $|X(R)|\leq d$, $\forall R\subseteq H$.
\end{proof}

\section{Clarkson's Algorithm Revisited}
Clarkson's algorithm can be used to compute a basis of some violator space $(H,\Violators)$, $n=|H|$.
It consists of two separate stages
and the Brute Force Algorithm (\stageIII). 
The results about the running time and the size of the sets involved is summarized in
Theorem \ref{theorem:runningtime1} and Theorem \ref{theorem:runningtime2}.

The main idea of both stages (\stageI{ }and \stageII) is the following:  We draw a random sample $R\subseteq H$ of
size $r=|R|$ and then compute a basis of $R$ using some other algorithm. The crucial point here is that $r\ll n$ hopefully. Obviously,
such an approach may fail to find a basis of $H$, and we might have to reconsider and enter a second round.
That is the point at which \stageI{ }and \stageII{ }most significantly differ.

In both stages we assume that the size of the ground set, i.e., $n$, is larger than $r$, such that we can
actually draw a sample of that size. We can assume this w.l.o.g., because it is easy to
incorporate an if statement at the beginning that directly calls the other algorithm should $n$ be too small. 

\subsection{The German Algorithm (\stageI)}
This algorithm works as follows. Let $(H,\Violators)$ be a violator space, $|H|=n$, and $\dim(H,\Violators) =d$.
 We draw a random sample $R\subseteq H$, $r=d\sqrt{n/2}$, only once, 
and initialize our working set $G$ with $R$. Then we enter a repeat loop, in which we compute a
basis $B$ of $G$ and check whether there are any violators in $H$. If no,
then we are done and return the basis $B$. If yes, then we add those violators to our working
set $G$ and repeat the procedure.

The analysis will show that $(i)$ the number of rounds is bounded by $d+1$, and $(ii)$
the size of $G$ in any round is bounded by $O(d\sqrt{n})$. See Theorem \ref{theorem:runningtime1}.

\begin{algorithm}[H]
    \SetKwInOut{Input}{input}
	\SetKwInOut{Output}{output}

	\Input{Violator space $(H,\Violators)$, $|H|=n$, and $\dim(H,\Violators)=d$}
	\Output{A basis $B$ of $(H,\Violators)$}
	\Indp
	\SetInd{0.5cm}{0em}
	\BlankLine
			$r \leftarrow d \sqrt{n/2}$\;
			Choose $R$ with $|R|=r$, $R\subseteq H$ u.a.r.\;
			$G \leftarrow R$\;
			
			\Repeat{$\Violators(B)=\emptyset$}{
				$B \leftarrow$ \stageII($G, \left.\Violators\right|_{G}$)\;
				$G \leftarrow G \cup \Violators(B)$\;
		}
		\Return{$B$}
	\BlankLine
    \caption{\stageI($H,\Violators$)\label{alg:stageI}}
    \end{algorithm}
    
    We will adopt some useful notations which we will use in the following proofs.
    First, let us point out that the notation $\left.\Violators\right|_{F}$ refers to the violator mapping
restricted to some set $F\subseteq H$.
    \begin{definition}
    \label{def:intermediatesets}
    For $i\geq 0$, by
    \[
    	B_{R}^{(i)}\text{, } V_{R}^{(i)}\text{, and } G_{R}^{(i)}
    \]
    we denote the sets $B$, $\Violators(B)$, and $G$ computed in round $i$ of the
    repeat loop above.
    Furthermore, we set $G_{R}^{(0)}:= R$, while $B_{R}^{(0)}$
    and $V_{R}^{(0)}$  are undefined. In particular, we have that
   $B_{R}^{(i)}$ is a basis of $G_{R}^{(i-1)}$, and  $V_{R}^{(i)} = \Violators(G_{R}^{(i-1)})$.
    If the algorithm performs exactly $\ell$ rounds, sets with indices $i> \ell$ are defined
    to be the corresponding sets in round $\ell$.
    \end{definition}

    The next one is an auxiliary lemma that we will need further on in the analysis. It is a 
generalization
    of the fact that there is at least one element of the basis of $H$ found as a violator in every 
round (see also
    Lemma \ref{lemma:53b}).
    \begin{lemma}
    \label{lemma:53}
    	For $j<i\leq \ell$, $B_{R}^{(i)}\cap V_{R}^{(j)} \neq \emptyset$.
    \end{lemma}
    \begin{proof}
    Assume that $B_{R}^{(i)}\cap V_{R}^{(j)} = \emptyset$. Together with consistency,
    $G_{R}^{(j-1)} \cap V_{R}^{(j)} = \emptyset$, this implies
    \[
    	(B_{R}^{(i)} \cup G_{R}^{(j-1)}) \cap V_{R}^{(j)} = \emptyset.
    \]
   Now, applying locality and the definition of basis, we get
   \begin{equation}
   \label{eq:27A}
   	\Violators(B_{R}^{(i)}\cup G_{R}^{(j-1)}) = V_{R}^{(j)} = \Violators(B_{R}^{(j)}).
   \end{equation}
   On the other  hand, since $V_{R}^{(i)} = \Violators(B_{R}^{(i)})$ and
   $B_{R}^{(i)}\subseteq B_{R}^{(i)} \cup G_{R}^{(j-1)}\subseteq G_{R}^{(i-1)}$,
   we can apply monotonicity and derive
   \begin{equation}
   \label{eq:27B}
   V_{R}^{(i)}=\Violators(B_{R}^{(i)}) = \Violators(B_{R}^{(i)} \cup G_{R}^{(j-1)}).
\end{equation}
   Note that $V(B_{R}^{(j)})\subseteq G_{R}^{(i-1)}$, because $G$ always contains the
   violators from previous rounds. Additionally, by equations (\ref{eq:27A}) and (\ref{eq:27B}) we have 
that
   $V_{R}^{(i)} = \Violators(B_{R}^{(i)}\cup G_R^{(j-1)}) = \Violators(B_{R}^{(j)})$.
   Thus, we can build a contradiction of consistency,
   \[
   	G_{R}^{(i-1)} \cap V_{R}^{(i)} \supseteq \Violators(B_{R}^{(j)}) \cap V_{R}^{(i)} = 
\Violators(B_{R}^{(j)}) \neq \emptyset.
   \]
   The last inequality holds because $j$ is not the last round.
    \end{proof}

    The following lemma is the crucial result that lets us interpret the development of the set $G$ in
    Algorithm \ref{alg:stageI} as a violator space itself.
    \begin{lemma}
    \label{lemma:54}
    Let $(H, \Violators)$ be a violator space of combinatorial dimension $d$. For any subset $R
\subseteq H$ define
    \[
    	\Gamma(R):= (V_{R}^{(1)}, \ldots, V_{R}^{(d)}).
    \]
    Using this we can define a new violator mapping as follows,
    \[
    	\Violators'(R):= \{ h\in H\backslash R \;|\; \Gamma(R)\neq \Gamma(R\cup \{h\} )\}.
    \]
    Then the following statements are true:
    \begin{enumerate}
    \item[\emph{(i)}] $(H, \Violators')$ is a violator space of combinatorial dimension at most ${d
+1}\choose {2}$.
    \item[\emph{(ii)}] The set $\Violators'(R)$ is given by
    \[
    	\Violators'(R) = V_{R}^{(1)} \cup \ldots \cup V_{R}^{(d)} = G_{R}^{(d)} \backslash R.
    \]
    \item[\emph{(iii)}] If $(H,\Violators)$ is nondegenerate, then so is $(H, \Violators')$.
    \end{enumerate}
    \end{lemma}
    
    In order to prove Lemma \ref{lemma:54} we first need an auxiliary claim. Note that $\dot{\cup}$
    denotes disjoint union.
    \begin{claim}
    \label{claim:1}
   Let $Q$ be any set with $Q = R\; \dot{\cup}\; T \subseteq H$ and $i < d$. If
   \begin{eqnarray*}
   	V_{Q}^{(j+1)} = V_{R}^{(j+1)} ,&& j\leq i,
   \end{eqnarray*}
   then
   \begin{eqnarray*}
   	G_{Q}^{(j)} = G_{R}^{(j)} \; \dot{\cup} \; T,& & j\leq i+1.
   \end{eqnarray*}
    \end{claim}
    \begin{proof}[Proof of Claim \ref{claim:1}]
    We prove the claim by induction on $i$. First, if $i=0$ the precondition reads $
    \Violators(Q) = \Violators(R)$. It follows that
    $G_{Q}^{(1)} = Q \cup \Violators(Q) = (R\;\dot{\cup}\; T) \cup \Violators(R) = G_{R}^{(1)}\;\dot{\cup}\; 
T$.

    Suppose the claim is true for $j\leq i$. From $V_{Q}^{(i+1)} = V_{R}^{(i+1)}$ we can deduce
    \[
    	G_{Q}^{(i+1)} = G_{Q}^{(i)} \cup V_{Q}^{(i+1)} = (G_{R}^{(i)} \;\dot{\cup}\; T) \cup V_{R}^{(i+1)} = 
G_{R}^{(i+1)} \;\dot{\cup}\; T.
    \]
    \end{proof}
    
    Before we proceed to the proof of Lemma \ref{lemma:54} let us first state the consequences, which
    we obtain by applying Lemma \ref{lemma:sampling} to the violator space that we constructed.
    
    \begin{theorem}[Theorem 5.5 of \cite{GWSampl01}]
    \label{theorem:sizeofG}
     For $R\subseteq H$ with $|H| = n$,  and a random sample of size $r$,
     \[
     	\Exp[|G_{R}^{(d)}|] \leq {{d+1}\choose{2}} \frac{n  - r}{r+1} + r.
     \]
     Choosing $r=d \sqrt{n/2}$ yields
     \[
     	\Exp[|G_{R}^{(d)}|] \leq 2(d+1)\sqrt{\frac{n}{2}}.
     \]
    \end{theorem}
    \begin{proof}[Proof of Theorem \ref{theorem:sizeofG}]
    The first inequality directly follows from the sampling lemma (Lemma \ref{lemma:sampling}),  
applied
    to the violator space $(H, \Violators')$, together with part $(ii)$ of Lemma \ref{lemma:54}. The 
second
    inequality follows from plugging in the value for $r$.
    \end{proof}
    
    Let us now come back to the Lemma.
    \begin{proof}[Proof of Lemma \ref{lemma:54}]
    ~\newline
     \textbf{Proof of (i).}
    We first need to check consistency and locality as defined in Definition \ref{def:violatorspace}.
    
    Consistency is easy, by the definition of $\Violators'$. Since the violators of $R\subseteq H$ are 
chosen
    from $H\backslash R$ exclusively, we can be sure that $R\cap \Violators'(R)=\emptyset$ for all $R
$.
    
    Let us recall what locality means. For sets $R\subseteq Q \subseteq H$, if $Q\cap\Violators'(R) = 
\emptyset$, then
    $\Violators'(Q) = \Violators'(R)$. This we are going to prove by induction on the size of $Q
\backslash R$.
    If $|Q\backslash R| = 0$, then the two sets are the same, and locality is obviously fulfilled. Now, 
suppose
    that $|Q\backslash R| = i$ and locality is true for any smaller value $j<i$. Consider some set $S$ 
fulfilling
    $R\subseteq S \subset Q$ and $Q=S\;\dot{\cup}\;\{q\}$. 
    First note that, if $Q\cap \Violators'(R) = \emptyset$, then also $S\cap \Violators'(R)= \emptyset$.
    Therefore, the precondition for the induction hypothesis is fulfilled, and we can conclude that
    $\Violators'(R) = \Violators'(S)$.  Bearing this in mind, we can make the following derivation:
    \[
    \begin{array}{lll}
    Q \cap \Violators'(R) = \emptyset &
    \Rightarrow &  Q \cap \Violators'(S) = \emptyset \\
    & \stackrel{q\in Q}{\Rightarrow} & q\not\in \Violators'(S) \\
    & \stackrel{\mbox{\tiny Def. (\ref{def:violator})}}{\Rightarrow} & \Gamma(S) = \Gamma(S\;\dot{\cup}\; \{q\}) = \Gamma(Q)\\
    & \stackrel{\mbox{\tiny Def. (\ref{def:violator})}}{\Rightarrow} & \Violators'(S) = \Violators'(Q) \\
    & \Rightarrow & \Violators'(R) = \Violators'(Q).
    \end{array}
    \]
    That shows the locality of the violator space $(H, \Violators')$.
    
    We still have to show that $(H, \Violators')$  has combinatorial dimension at most ${d+1 \choose 
2}$.
    To this end we prove that $\Violators'(B_{R}) = \Violators'(R)$, where
    \[
    	B_{R}:= R\cap \bigcup_{i=1}^{d} B_{R}^{(i)}.
    \]
     Note that $B_{R}$, as we will show in {(iii)}, is in fact the unique basis of the set $R\subseteq H$.
     By bounding the size of $B_R$ we therefore bound the combinatorial dimension of
    $(H, \Violators')$. 
     Equivalent to  $\Violators'(B_{R}) = \Violators'(R)$ we show that $V_{B_{R}}^{(j)} = V_{R}^{(j)}$, for $1\leq j \leq d$, using induction on $j
$. For
    $j=1$ we get
    \[
    	\Violators(R) = \Violators(B_{R}\cup (R\backslash B_{R})) = \Violators(B_{R}),
    \]
    because $R\backslash B_{R}$ is disjoint from $B_{R}^{(1)}$, the basis of $R$. Therefore,
    $R\backslash B_{R}=R\,\backslash  \bigcup_{i=1}^{d} B_{R}^{(i)}$
    can be removed from $R$ without changing the set of violators.
    
    Now assume that the statement holds for $j\leq d-1$ and consider the case $j=d$. By Claim 
    \ref{claim:1}, we get $G_{R}^{(j-1)} = G_{B_{R}}^{(j-1)} \;\dot{\cup}\; (R\backslash B_{R})$.
    Since $R\backslash B_{R}$ is disjoint from the basis $B_{R}^{(j)}$ of $G_{R}^{(j-1)}$ it follows that
    \[
    	V_{R}^{(j)} = \Violators(G_{R}^{(j-1)}) = \Violators(G_{B_{R}}^{(j-1)} \;\dot{\cup}\; (R\backslash 
	B_{R})) =	\Violators(G_{B_{R}}^{(j-1)}) = V_{B_{R}}^{(j)}.
    \]
    
    To bound the size of $B_{R}$, we observe that
    \[
    	|R\cap B_{R}^{(i)}| \leq d+1-i,
    \]
    for all $i\leq \ell$ (the number of rounds in which $\Violators(B)\neq\emptyset$). This follows
    from Lemma \ref{lemma:53}. $B_{R}^{(i)}$ has at least one element in each of the $i-1$ sets
    $V_{R}^{1},\ldots, V_{R}^{(i-1)}$, which are in turn disjoint from $R$. Hence we get
    \[
    	|B_{R}| \leq \sum_{i=1}^{\ell} |R\cap B_{R}^{(i)}| \leq {d+1 \choose 2}.
    \]
    
    ~\newline
     \textbf{Proof of (ii).}
     We show that if some constraint $q\in H$ is in $\Violators'(R)$ then it is also in $V_R^{(i)}$ for
     some $1\leq i \leq d$. On the other hand if $q\not\in \Violators'(R)$ then $q$ is not in
     any of the $V_{R}^{(i)}$, $1\leq i \leq d$. This proves the statement of $(ii)$.
     
     Assume $q\in\Violators'(R)$ and let $Q:=R\cup \{q\}$. Consider the largest index
     $i< d-1$ such that
     \[
     	V_{R}^{(j+1)} = V_{Q}^{(j+1)},\; j\leq i.
     \]
     Note that such an index $i$ must exist, because $\Violators'(R) \neq \Violators'(Q)$, which
     simply follows from $q\in \Violators'(R)$ and $q\not\in\Violators'(Q)$.
     Then, from Claim \ref{claim:1} it follows that $G_{Q}^{(i+1)}= G_{R}^{(i+1)} \;\dot{\cup}\; \{q\}$, and
     by assumption on $i$ we know that $V_{R}^{(i+2)} \neq V_{Q}^{(i+2)}$. Therefore, by the contrapositive of locality,
     we conclude $(G_{R}^{(i+1)} \;\dot{\cup}\; \{q\}) \cap \Violators(G_{R}^{{(i+1)}})\neq \emptyset$.
     This means that $q\in \Violators(G_{R}^{{(i+1)}})= V_{R}^{(i+2)}$, because otherwise the 
consistency of $G_{R}^{(i+1)}$
     would be violated.
     
     On the other hand, if $q\not\in \Violators'(R)$, then $\Violators'(R) = \Violators'(Q)$, or equivalently
     $V_{R}^{(i)} = V_{Q}^{(i)}$, for $1\leq i \leq d$.
    However, because $(H, \Violators)$ is consistent it follows that $q\not\in V_{Q}^{(i)}$, and
    therefore $q\not\in V_{R}^{(i)}$,  for $1\leq i\leq d$.
    
    ~\newline
     \textbf{Proof of (iii).}
     Nondegeneracy of $(H,\Violators')$ follows if we can show that every set $R\subseteq H$
     has the set $B_{R}$ as its unique basis. To this end we prove that whenever we have
     $L\subseteq R$ with $\Violators'(L) = \Violators'(R)$, then $B_{R}\subseteq L$.
     
     Fix $L\subseteq R$ with $\Violators'(L) = \Violators'(R)$, i.e.,
     \[
     	 V_{R}^{(i)} = V_{L}^{(i)}, \; 1\leq i\leq d.
     \]
     Claim \ref{claim:1} then implies
     \[
     	G_{R}^{(i)} = G_{L}^{(i)} \;\dot{\cup}\; (R\backslash L), \; 0\leq i \leq d,
     \]
     and the nondegeneracy of $(H, \Violators)$ yields that $G_{R}^{(i)}$ and
     $G_{L}^{(i)}$ have the same unique basis $B_{R}^{(i+1)}$, for all $0\leq i \leq d$.
     Note that $B_{R}^{(i+1)}$ is indeed contained in $G_{L}^{(i)}$, because 
     $\Violators(G_L^{(i)})=\Violators(G_R^{(i)})=\Violators(G_L^{(i)}\dot{\cup}(R\backslash L))=
     \Violators(B_{R}^{(i+1)})$ for $0\leq i \leq d$. That means, if there exists a basis of $G_L^{(i)}$,
     that by definition would also be a basis
     of $G_R^{(i)}$, but distinct from $B_{R}^{(i+1)}$, nondgeneracy is violated.
     
     It follows that $G_{L}^{(d-1)}$ contains
     \[
     	\bigcup_{i=1}^{d} B_{R}^{(i)},
     \]
     so $L$ contains
     \[
     	L\cap \bigcup_{i=1}^{d} B_{R}^{(i)} = R\cap \bigcup_{i=1}^{d} B_{R}^{(i)}.
     \]
     The latter equality holds because $R\backslash L$ is disjoint from $G_{L}^{(d)}$,
     thus in particular from the union of the $B_{R}^{(i)}$.
     \end{proof}
    
    \begin{theorem}
    \label{theorem:runningtime1}
    Let  $(H,\Violators)$ be a violator space of combinatorial dimension $d$, and  $n= |H| $.
Then the algorithm \stageI~computes a basis of $(H,\Violators)$ with
 at most $d+1$ calls to \stageII, with an expected number of 
at most $O(d\sqrt{n})$
constraints each.
    \end{theorem}
    \begin{proof}
    According to Lemma \ref{lemma:53} (and maybe more intuitively according to Lemma \ref{lemma:53b}),
    in every round except the last one we add at least one element of any basis of $(H,\Violators)$ to $G$.
    Since the size of the basis is bounded by $d$ we get that the number of rounds is at most $d+1$. Furthermore,
    according to Theorem \ref{theorem:sizeofG}, and our choice $r=d\sqrt{n/2}$, the expected size of $G$ will not
    exceed $2(d+1)\sqrt{{n}/{2}}$ in any round.
    \end{proof}

\subsection{The Swiss Algorithm (\stageII)}
\label{subsec:secondstage}

The algorithm \stageII~proceeds similar as the first one.
Let the input be a violator space $(H,\Violators)$, $|H|=n$, and $\dim(H,\Violators)=d$.

First, let us 
(re)introduce the notation $R^{(i)}$, $B^{(i)}$, and $V^{(i)}$ for $i\geq 1$, similar as in 
Definition \ref{def:intermediatesets}, for the sets $R$, $B$ and $\Violators(R)$ of round $i$
respectively. The set $B^{(i)}$ is a basis of $R^{(i)}$ and $V^{(i)}=\Violators(R^{(i)})=
\Violators(B^{(i)})$.
Since we draw a random sample in every round it does not make sense to index the sets
$B^{(i)}$ and $V^{(i)}$ by $R$, so we drop this subscript.

After the initialization we enter the first round and choose a random sample $R^{(1)}$ of size
$r=2 d^2$ uniformly at random from $H$. Then we compute an intermediate basis $B^{(1)}$
of the violator space  $(R^{(1)},\left.\Violators\right|_{R^{(1)}})$ by using \stageIII~as
a black box. 
In the next step we compute the set of violated constraints, i.e., $V^{(1)}$.
So far, it is the same thing as the first stage. But now, instead of enforcing the violated
constraints by adding them to the active set, we increase the probability that the violated
constraints are chosen in the next round. This is achieved by means of the multiplicity
or weight variable $\mu$.

    \begin{definition}
\label{def:intermediatemultiplicities}
With every $h\in H$ we associate the \emph{multiplicity} $\mu_h\in\mathbb{N}$. For an arbitrary
set $F\subseteq H$ we define the \emph{cumulative multiplicity} as
\[
	\mu(F):= \sum_{h\in F} \mu_h.
\]
For the analysis we also need to keep track of this value across different iterations
of the algorithm. For $i\geq 0$ we will use $\mu^{(i)}_h$ (and
$\mu^{(i)}(F)$) to denote the (cumulative) multiplicity at the end of round $i$.
We define $\mu_h^{(0)}:= 1$ for any $h \in H$,
and therefore $\mu^{(0)}(F)=|F|$.
\end{definition}

Now back to the algorithm. To increase the probability that a constraint $h\in V^{(i)}$ is chosen in the
random sample of round $i+1$ we double the multiplicity of $h$,
i.e., $\mu^{(i)}_h = 2\mu^{(i-1)}_h$.

The multiplicities determine
how the random sample $R^{(i+1)}$ is chosen.
To this end we construct a multiset
$\hat{H}^{(i+1)}$ to which we add $\mu^{(i)}_h$ copies of every element $h\in H$.
To simplify notation, let us for a moment fix the round $i+1$ and drop the
corresponding superscript.

We define the function $\phi: 2^H \rightarrow 2^{\hat{H}}$ as
the function that maps a set of elements from $H$ to the set of corresponding elements in $\hat{H}$, 
i.e., for
$F\subseteq H$ 
\[
	\phi(F):= \bigcup_{h\in F}\{h_1,\ldots,h_{\mu_h}\},
\] 
where the $h_j$, $1\leq j \leq \mu_h$, are the distinct copies of $h$. For example, $\hat{H} = \phi(H)$.
Conversely, let $\psi: 2^{\hat{H}}\rightarrow 2^{H}$ be the  function that collapses a given subset
of $\hat{H}$ to their original elements in $H$, i.e., for $\hat{F}\subseteq \hat{H}$,
\[
	\psi(\hat{F}):= \{h\in H \;|\; \phi(\{h\})\cap \hat{F}\neq\emptyset\}.
\]

Reintroducing the superscript $i+1$ we can simply say that we construct $\hat{H}^{(i+1)}=\phi(H)$
using the multiplicities from round $i$.
The sample $\hat{R}^{(i+1)}$ is then chosen u.a.r. from the $r$-subsets of
$\hat{H}^{(i+1)}$. In the following the multiset property will not be important any more and we
can discard multiple entries to obtain $R^{(i+1)}=\psi(\hat{R}^{(i+1)})$. Note that $1\leq |R^{(i+1)}|\leq 
r$.
Then we continue as in round $1$. Note that in the first round this is in fact equivalent
to choosing an $r$-subset u.a.r. from $H$, because $\mu^{(0)}_h=1$ for all $h\in H$.

The algorithm terminates as soon as $V^{(\ell)}=\emptyset$ for some round $\ell\geq 1$ and
returns the basis $B^{(\ell)}$.

{
\parskip2em
\begin{algorithm}[H]
    \SetKwInOut{Input}{input}
	\SetKwInOut{Output}{output}
	\Input{Violator space $(H,\Violators)$, $|H|=n$, and $\dim(H,\Violators)=d$}
	\Output{A basis $B$ of $(H,\Violators)$}
	\Indp
	\SetInd{0.5cm}{0em}
	\BlankLine
			$\mu_h \leftarrow 1$ for all $h\in H$\;
			$r \leftarrow 2 d^2$\;
			\Repeat{$V(B)=\emptyset$}{
				choose random $R$ from $H$ according to $\mu$\;
				$B \leftarrow$ \stageIII($R, \left.\Violators\right|_{R}$)\;
				$\mu_h \leftarrow 2\mu_h$ for all $h\in \Violators(B)$\;
			}
		\Return{$B$}
	\BlankLine
    \caption{\stageII($H,\Violators$)\label{stageII}}
    \end{algorithm}

Let us first discuss an auxiliary lemma very similar in flavour to Lemma \ref{lemma:53}.
}

\begin{lemma}[Observation 22,  \cite{journals/dam/GartnerMRS08}]
\label{lemma:53b}
Let $(H, \Violators)$ be a violator space, $F\subseteq G \subseteq H$, and $G\cap\Violators(F)\neq 
\emptyset$.
Then $G\cap \Violators(F)$ contains at least one element from every basis of $G$.
\end{lemma}
\begin{proof}
Since the proof is pretty short we repeat it here.
Let $B$ be some basis of $G$ and assume that $B\cap G \cap \Violators(F) = B\cap \Violators(F) = 
\emptyset$. From consistency we get
$F\cap \Violators(F)= \emptyset$. Together this implies
\[
	(B\cup F) \cap \Violators(F) = \emptyset.
\]
Applying locality and monotonicity, we get
\[
	\Violators(F) = \Violators(B\cup F) = \Violators(G),
\]
meaning that $G\cap\Violators(G)= G\cap\Violators(F)=\emptyset$, a contradiction.
\end{proof}

The analysis of \stageII~will show that the elements in any basis $B$ of
$H$ will increase their multiplicity so quickly
that they are chosen with high probability after a logarithmic number
of rounds. This, of course, means that the algorithm will terminate, because there will be no violators.
Formally, we will have to employ a little trick though. We will consider a modification
of \stageII~that runs forever, regardless of the current set of violators! Let us call the modified algorithm \stageIIforever. We call a particular round $i$ \emph{controversial} if
$V^{(i)}\neq\emptyset$. Furthermore, let $C_\ell$ be the event that the first  $\ell$ rounds are controversial
 in \stageIIforever.

\begin{lemma}
\label{lemma:lowerbound}
Let $(H,\Violators)$ be a violator space,  $|H|=n$, $\dim{(H,\Violators)}=d$,
$B$ any basis of $H$, and $k \in \mathbb{N}$ some positive integer.
Then, in \stageIIforever, the following holds for the expected cumulative
multiplicity of $B$ after $kd$ rounds,
\[
	2^k \Pr[C_{kd}] \leq \Exp[\mu^{(kd)}(B)] .
\]
\end{lemma}
\begin{proof}
In any controversial round,
Lemma \ref{lemma:53b} asserts that $B\cap V^{(i)}\neq\emptyset$. So, in every controversial round, the 
multiplicity of at least one element in $B$ is doubled. Therefore, by conditioning on the event
that the first $kd$ rounds are controversial,
there must be a constraint in $B$ that has been doubled at least $k$ times (recall that 
$|B|\leq d$). It follows that 
$\Exp[\mu^{(k d)}(B)] = \Exp[\mu^{(k d)}(B)\,|\,C_{kd}]\Pr[C_{kd}] +
\Exp[\mu^{(k d)}(B)\,|\,\overline{C_{kd}}]\Pr[\overline{C_{kd}}]\geq 2^k \Pr[C_{kd}]$.
\end{proof}

\begin{lemma}
\label{lemma:upperbound}
Let $(H,\Violators)$ be a violator space, $|H|=n$, $\dim{(H,\Violators)}=d$,
$B$ any basis of $H$, and $k \in \mathbb{N}$ some positive integer.
Then, in \stageIIforever, the following holds for the expected cumulative
multiplicity of $B$ after $kd$ rounds,
\[
	\Exp[\mu^{(kd)}(B)] \leq n\left(1+\frac{d}{r}\right)^{kd}.
\]
\end{lemma}
\begin{proof}

Let us point out first, that the following analysis goes through for
\stageIIforever{ }as well
as for \stageII, but to make it match Lemma \ref{lemma:lowerbound} we formulated it using
the former.

Note that $\Exp[\mu^{(kd)}(B)] \leq \Exp[\mu^{(kd)}(H)]$, because $B
\subseteq H$. Therefore, if we show the upper bound for the latter expectation we are done. Let  
$\ell:=k d$ be the number of rounds, and
$\Delta^{(i)}(F) := \mu^{(i)}(F) - \mu^{(i-1)}(F)$ the increase of multiplicity from one round to another, for any $i\geq1$ and
$F\subseteq H$.
We write the expected weight of $H$ after $\ell$ rounds as the sum of the initial 
weight plus the expected increase in weight in every round from $1$ to $\ell$,
\begin{equation}
\label{eq:d01}
	\Exp[\mu^{(\ell)}(H)]  =  \Exp[\mu^{(0)}(H)] + \sum_{i=1}^\ell \Exp[\Delta^{(i)}(H)].
\end{equation}
The first term is easy, $\Exp[\mu^{(0)}(H)]=n$, and the second term we write as a conditional 
expectation,
assuming that the weight in round $i-1$ was $t$,
\begin{equation}
\label{eq:d02}
	\sum_{i=1}^\ell \Exp[\Delta^{(i)}(H)] = \sum_{i=1}^{\ell} \sum_{t=0}^{\infty} 
	\Exp[\Delta^{(i)}(H)|\mu^{(i-1)}(H)=t]
	\Pr[\mu^{(i-1)}(H)=t].
\end{equation}

Now comes the crucial step. According to Lemma \ref{lemma:sampling} we can 
upper bound
$\Exp[\Delta^{(i)}(H)|\mu^{(i-1)}(H)=t]$
by interpreting it as the expected 
number of violators of a
multiset extension of $(H,\Violators)$.
To this end we construct a violator space
$(\hat{H}^{(i)},\hat{\Violators})$, where $\hat{H}^{(i)}= \phi(H)$
using the multiplicities from round $i-1$.
Let us fix round $i$ and drop the superscript for the moment. For any
$\hat{F}\subseteq \hat{H}$ we define
\[
	\hat{\Violators}(\hat{F}):= \phi(\Violators({\psi(\hat{F})})).
\]
We observe that $(\hat{H},\hat{\Violators})$ is indeed a violator space. For
$\hat{F}\subseteq \hat{H}$, consistency
is preserved, because from consistency of $(H,\Violators)$ it follows that 
$\phi(\psi(\hat{F}))\cap\phi(\Violators(\psi(\hat{F})))=\emptyset$, and knowing
$\hat{F}\subseteq\phi(\psi(\hat{F}))$, we can conclude consistency of $(\hat{H}, \hat{\Violators})$. 
Similarly,
for $\hat{F}\subseteq\hat{G}\subseteq\hat{H}$, locality of $(H,\Violators)$ tells us that if
$\phi(\psi(\hat{G}))\cap\phi(\Violators(\psi(\hat{F})))=\emptyset$ then
$\phi(\Violators(\psi(\hat{F})))=\phi(\Violators(\psi(\hat{G})))$, and knowing $\hat{G}\subseteq
\phi(\psi(\hat{G}))$,
locality of $(\hat{H},\hat{\Violators})$ follows.

The violator space we just constructed has the same ground set $\hat{H}$ by means of which we 
draw the random
sample $R$ in every round. By supplying a valid violator mapping we asserted that we can 
apply
the sampling lemma to that process.
Some thinking reveals that $d = \dim(H,\Violators)=\dim(\hat{H},\hat{\Violators})$ (even though we 
introduced
degeneracy), and we can conclude
\begin{equation}
\label{eq:d03}
	\Exp[\Delta^{(i)}(H)|\mu^{(i-1)}(H)=t] =  \Exp[|\hat{\Violators}(\hat{R}^{(i)})|] \leq d \frac{t-r}{r+1} .
\end{equation}
Therefore we get the simplified expression

\begin{center}
\begin{math}
\displaystyle
\renewcommand{\arraystretch}{2.8}
\begin{array}{lll}
	\Exp[\mu^{(\ell)}(H)] 
	& \leq & \displaystyle n + \sum_{i=1}^\ell \sum_{t=0}^{\infty} 
	d \frac{t-r}{r+1}
	\Pr[\mu^{(i-1)}(H)=t] \label{eq:d04}\\ 
	& = & \displaystyle n + \sum_{i=1}^\ell \left(\frac{d}{r+1}
	 \sum_{t=0}^{\infty}  t  \Pr[\mu^{(i-1)}(H)=t]\right.\\
	&&\displaystyle \left. - \frac{d r}{r+1}\sum_{t=0}^{\infty}
	\Pr[\mu^{(i-1)}(H)=t]\right) \label{eq:d05}\\
	& = & \displaystyle n + \frac{d}{r+1}\sum_{i=1}^\ell
	\Exp[\mu^{(i-1)}(H)] - \ell  \frac{d r}{r+1}.  \label{eq:d06}
\end{array}
\end{math}
\end{center}

The first line is derived from (\ref{eq:d01}), (\ref{eq:d02}), and (\ref{eq:d03}). The rest is routine. 
Dropping the last term  we get the following recursive equation,
\[
	\Exp[\mu^{(\ell)}(H)] \leq
	n + \frac{d}{r+1}\sum_{i=0}^{\ell-1} \Exp[\mu^{(i)}(H)],
\]
which easily resolves to the claimed bound.
\end{proof}

Using $\ell = kd$, and combining Lemmata \ref{lemma:lowerbound} and  \ref{lemma:upperbound}, 
we now know that
\[
2^{k}~\Pr[C_{\ell}] \leq n\left(1+\frac{d}{r}\right)^{\ell}.
\]

This inequality gives us a useful upper bound on $\Pr[C_\ell]$, 
because the left-hand side power grows faster than the right-hand side 
power as a function of $\ell$, given that $r$ is chosen large enough. 

Let us choose
$r=c\,d^2$ for some constant $c>\log_2e\approx 1.44$.
We obtain
\[\Pr[C_\ell] \leq n \left(1+\frac{1}{c\,d}\right)^\ell /\, 2^{k}
\leq n\, 2^{(\ell\log_2e)/(c\,d)-k},
\]
using  $1+x\leq e^x=2^{x\log_2e}$ for all $x$. This further gives us
\begin{equation}
\Pr[C_\ell] \leq n \alpha^\ell,
\label{eq:C_k_up}
\end{equation}
\[\alpha = \alpha(d,c) = 2^{(\log_2e - c)/(c\,d)} < 1.\]
This implies the following tail estimate.
\begin{lemma}
  For any $\beta>1$, the probability that \stageIIforever~
  starts with at least $\lceil\beta\log_{1/\alpha} n\rceil$ controversial 
  rounds is at most
\[
n^{1-\beta}.
\]   
\end{lemma}

\begin{proof}
  The probability for at least this many leading controversial rounds 
  is at most
\[
\Pr[C_{\lceil\beta\log_{1/\alpha}n\rceil}]\leq n \alpha^{\lceil\beta\log_{1/\alpha}n\rceil}
\leq n \alpha^{\beta\log_{1/\alpha}n}
= n n^{-\beta} = n^{1-\beta}.
\]
\end{proof}

We can also bound the expected number of
leading controversial rounds in \stageIIforever, and this
bounds the expected number of rounds in \stageII, because \stageII~terminates upon 
the first non-controversial round it encounters.

\begin{theorem}
\label{theorem:runningtime2}
Let $(H,\Violators)$ be a violator space, $|H|=n$, and $\dim{(H,\Violators)}=d$.
Then the algorithm \stageII~computes a basis of $H$ with
 an expected number of at most $O(d \ln n)$ calls to \stageIII, with 
at most $O(d^2)$
constraints each.
\end{theorem}
\begin{proof}
By definition  of $C_\ell$, the expected number of leading controversial rounds
in \stageIIforever{ }is 
\[
\sum_{\ell\geq 1}\Pr[C_\ell].
\]
For any $\beta>1$, we can use (\ref{eq:C_k_up}) to bound this by

\begin{math}
\renewcommand{\arraystretch}{2.15}
\begin{array}{lll}
\displaystyle \sum_{\ell=1}^{\lceil\beta\log_{1/\alpha}n\rceil-1}1 +
n \sum_{\ell=\lceil\beta\log_{1/\alpha}n\rceil}^{\infty}\alpha^\ell
&=&\lceil\beta\log_{1/\alpha}n\rceil-1 + 
n\frac{\alpha^{\lceil\beta\log_{1/\alpha}n\rceil}} {1-\alpha}\\ 
&\leq& \beta\log_{1/\alpha}n + \frac{n^{1-\beta}}{1-\alpha} \\
&=& \beta\log_{1/\alpha}n + o(1). 
\end{array}
\end{math}

This upper bounds the expected number of rounds in \stageII. In every round of \stageII~one call to
\stageIII~is made, using $c\,d^2$ constraints, where $c>\log_2e$
is constant.
\end{proof}

\section{Hypercube Partitions}
Let $H$ be a finite set. Consider the graph on the vertices $2^H$, where two vertices $F, G$
are connected by an edge if they differ in exactly one element, i.e., $G=F\,\dot{\cup}\,\{h\}$, $h\in H$.
This graph is a hypercube of dimension $n=|H|$. For
the sets $A \subseteq  B\subseteq H$, we define 
 $[A, B]:= \{C\subseteq H \;|\; A\subseteq C \subseteq B\}$ and call any such
 $[A,B]$ an \emph{interval}. A \emph{hypercube  partition} is a partition $\mathcal{P}$
 of $2^H$ into (disjoint) intervals.
 
 Let $(H,\Violators)$ be a violator space. We call two sets $F,G\subseteq H$ \emph{equivalent}
 if $\Violators(F)=\Violators(G)$, and let $\mathcal{H}$ be the partition of $2^H$ into
 equivalence classes w.r.t. this relation. We call $\mathcal{H}$ the \emph{violation pattern}
 of $(H,\Violators)$.

Before we formulate and prove the Hypercube Partition Theorem, we need to introduce some 
notation. We extend the notion of violator spaces by the concept of \emph{anti-basis}.

\begin{definition}
Consider a violator space $(H,V)$. We say that $\bar{B}\subseteq H$ is an \emph{anti-basis} if
for all proper supersets $F \supset \bar{B}$ we have $\Violators(\bar{B})\cap F \neq \emptyset$.
An anti-basis of $G\subseteq H$ is a maximal superset $\bar{B}$ of $G$ with $\Violators(\bar{B})=
\Violators(G)$. 
 \end{definition}
 
 Note that a maximal superset $\bar{B}$ of $G$ such that $\Violators(\bar{B})=\Violators(G)$ is 
indeed 
 an anti-basis of $G$. Suppose that there is a set $\bar{B}' \supset \bar{B}$ with
 $\Violators(\bar{B})\cap\bar{B}'=\emptyset$. Locality then decrees that $\Violators(\bar{B})=
\Violators(\bar{B}')$, but this contradicts the maximality of $\bar{B}$.
 
 \begin{lemma}
 \label{lemma:uniqueantibasis}
Consider  the violator space $(H,\Violators)$. For any $G\subseteq H$ there is a unique
 anti-basis $\bar{B}_G$ of $G$.  
 \end{lemma}

 \begin{proof}
 Suppose that there exist two distinct anti-bases $\bar{B}$ and $\bar{B}'$ of G.
 Because of $\Violators(\bar{B})=\Violators(\bar{B}')$ and consistency we have that
 $(\bar{B}\cup\bar{B}')\cap \Violators(\bar{B}) = (\bar{B}\cup\bar{B}')\cap \Violators(\bar{B}') = 
\emptyset$.
 Therefore, by locality, $\Violators(\bar{B}\cup\bar{B}') = \Violators(\bar{B}') = \Violators(\bar{B})$.
 Since $\bar{B}$ and $\bar{B}'$ are distinct, it cannot be that
  $\bar{B}\backslash\bar{B}'= \emptyset$ and $\bar{B}'\backslash\bar{B} = \emptyset$
  at the same time. Then, in any case, $|\bar{B}\cup\bar{B}'|>|\bar{B}|$ or $|\bar{B}\cup\bar{B}'|>|
\bar{B}'|$
  holds, which contradicts the maximality of the anti-bases.
   \end{proof}
   
   \begin{corollary}
   Let $(H,\Violators)$ be a violator space, $G\subseteq H$, $B_G$ any basis of $G$, and $\bar{B}_G
$
   the unique anti-basis of $G$. Then for any set $F$, $B_G\subseteq F \subseteq \bar{B}_G$, $F$
   and $G$ are equivalent, i.e., $\Violators(F)= \Violators(G)$.
   \end{corollary}
   \begin{proof}
   This is an immediate consequence of monotonicity (Lemma \ref{lemma:monotonicity}).
   \end{proof}

\begin{lemma}
\label{lemma:patterndetspace}
${\cal H}$ completely determines $(H,\Violators)$.
\end{lemma}

\begin{proof} Let $G\subseteq H$. There is a unique anti-basis
$\overline{B}_G$ of $G$, meaning that in ${\cal H}$, there is a unique
inclusion-maximal superset of $G$ in the same class of the partition.
This implies that $\Violators(G)=\Violators(\overline{B}_G=H\setminus\overline{B}_G)$, so
$(H,\Violators)$ is reconstructible from ${\cal H}$.
\end{proof}

\begin{lemma}
\label{lemma:nondegdetpartition}
If $(H,\Violators)$ is nondegenerate (unique bases), then ${\cal H}$ is a hypercube
partition. 
\end{lemma}

\begin{proof}
We first show that $\Violators(B)=\Violators(B')$ implies $\Violators(B\cap B')=\Violators(B\cup
B')=\Violators(B)$. The latter has been shown for the existence of a unique
anti-basis. For the former, we argue as follows. Let $A$ be the unique
basis of $B\cup B'$. Then $\Violators(A)=\Violators(B)=\Violators(B')$. But then $A$ is also the
unique basis of $B$ and $B'$. It follows that $A\subseteq B\cap B'$,
and by locality we get $\Violators(A)=\Violators(B\cap B')=\Violators(B)$.

This argument implies that any partition class ${\cal C}$ is contained
in the interval
$[\bigcap_{C\in{\cal C}}C, \bigcup_{C\in{\cal C}}C]$.
On the other hand, the whole interval is contained in ${\cal C}$ by locality,
so we are done. 
\end{proof}

Lemma \ref{lemma:patterndetspace} and  \ref{lemma:nondegdetpartition}
together imply that there is an injective mapping 
from the set of nondegenerate violator spaces to the set of hypercube
partitions. It remains to show that the mapping is surjective.

\begin{theorem}
Any hypercube partition ${\cal P}$ is the violation pattern
of some nondegenerate violator space $(H,\Violators)$
 \end{theorem}
\begin{proof}
Let $G\subseteq H$, and let $[B,B']$ be the interval containing
$G$. We define $\Violators(G)=H\setminus B'$ and claim that this is a
nondegenerate violator space with violation pattern ${\cal P}$. The
latter is clear, since $\Violators(F)=\Violators(G)$ if and only if $F,G\subseteq
[B,B']$. To see the former, we observe that consistency holds because
of $G\subseteq B'$. To prove locality, choose $G\subseteq G'$ with
$H\setminus B'=\Violators(G)\cap G'=\emptyset$. In particular, $G'\subseteq
B'$, so $G'$ is also in $[B,B']$ and we get $\Violators(G)=\Violators(G')$ by definition
of $\Violators$.

It remains to show that the violator space thus defined is nondegenerate.
Let $B,B'$ be two sets with $\Violators(B)=\Violators(B')$, meaning that they are in the
same partition class of ${\cal P}$. But then $B\cap B'$ is also in the
same class, and we get $\Violators(B)=\Violators(B\cap B')$. This implies existence of
unique bases.
\end{proof}

\section{Conclusion}
We analyzed Clarkson's algorithm
in what we believe to be its most general as well as natural
setting. Additionally, we have given the equivalence between non-degenerate violator
spaces and hypercube partitions, which could help 
identifying further applications in computational geometry as well as other fields
of computer science. Another major challenge will be to establish
a subexponential analysis for the third stage, \stageIII, in the framework of violator spaces
(as there already exists for LP's and LP-type problems), in order
to get stronger bounds when the dimension is only moderately small.

\newpage
\bibliographystyle{plain}

\end{document}